\documentclass[twocolumn,
showpacs,preprintnumbers,amsmath,amssymb]{revtex4}

\usepackage{amssymb,amsmath,amsthm}
\usepackage{graphicx}
\newtheorem{Thm}{Theorem}

\newtheorem{Lem}{Lemma}
\theoremstyle{definition}

\newcommand{\bra}[1]{{\left\langle #1 \right|}}
\newcommand{\ket}[1]{{\left| #1 \right\rangle}}

\newcommand{\T}{\mbox{$\mathrm{tr}$}}

\begin{document}
\title{Unification of multi-qubit polygamy inequalities}

\author{Jeong San Kim}
\email{freddie1@suwon.ac.kr} \affiliation{
 Department of Mathematics, University of Suwon, Kyungki-do 445-743, Korea
}
\date{\today}

\begin{abstract}
We establish a unified view of polygamy of multi-qubit entanglement.
We first introduce a two-parameter generalization of entanglement of assistance
namely unified entanglement of assistance for bipartite quantum states, and
provide an analytic lowerbound in two-qubit systems.
We show a broad class of polygamy inequalities of multi-qubit entanglement in terms of
unified entanglement of assistance that encapsulates all known multi-qubit polygamy inequalities as special cases.
We further show that this class of polygamy inequalities can be improved into tighter inequalities
for three-qubit systems.
\end{abstract}

\maketitle
\section{Introduction}

As a quantum correlation among different systems, quantum entanglement
shows an essential difference from classical correlations.
If a pair of parties in a multi-party quantum system are
maximally entangled then they cannot share any entanglement~\cite{ckw,ov} nor classical
correlations~\cite{kw} with the rest of the system.
This restricted shareability of entanglement in multi-party quantum systems is known as {\em
monogamy of entanglement}~(MoE)~\cite{T04}; more entanglement shared between two parties necessarily
implies less entanglement shared with the rest of the system.
Furthermore, shared entanglement between two parties even limits the amount of classical correlation that can
be shared with the other parties.

MoE plays a crucial role in many quantum information processing tasks.
In quantum key-distribution protocols, the possible amount of information an eavesdropper could obtain about
the secret key can be restricted by MoE, which is
the fundamental concept of security proof.
MoE also plays an important role
in condensed-matter physics such as
the $N$-representability problem for
fermions~\cite{anti}.

The first characterization of MoE was proposed as an inequality in three-qubit systems~\cite{ckw}
using concurrence~\cite{ww} to quantify shared bipartite entanglement. Later, monogamy inequality was
generalized into multi-qubit systems in terms of various entanglement measures~\cite{ov, KSRenyi, KT, KSU},
and also some cases of higher-dimensional quantum systems rather than qubits~\cite{kds}.

Whereas, monogamy inequality is about the restricted shareability of multipartite entanglement,
the dual concept of the sharable entanglement, namely distributed entanglement, is known to have a
polygamous property in multipartite quantum systems.
A mathematical characterization for the {\em polygamy of entanglement} was first provided for multi-qubit
systems~\cite{gbs} using concurrence of assistance (CoA)~\cite{lve} to quantify the distributed bipartite entanglement.
Recently, a broad class of polygamy inequalities for multi-qubit systems was proposed~\cite{KT}, and a
polygamy inequality in tripartite quantum systems of arbitrary dimension was also shown using
entanglement of assistance (EoA)~\cite{BGK}.

Here, we provide a unified view of these polygamy inequalities of multi-qubit entanglement.
We first introduce a two-parameter generalization of
EoA namely {\em unified entanglement of assistance} (UEoA)
for bipartite quantum states, and provide an analytic lower bound for UEoA
in two-qubit systems. By investigating the functional relation between UEoA and concurrence, we establish
a two-parameter class of polygamy inequalities of multi-qubit entanglement in terms of UEoA.
This new class of polygamy inequalities reduces to every known  multi-qubit polygamy inequalities
as special cases, therefore our new class of polygamy inequalities also provides an interpolation among
various polygamy inequalities of multi-qubit entanglement.
We further show that our polygamy inequality can be improved into a
tighter inequality for three-qubit pure states.

This paper is organized as follows. In Section~\ref{Subsec:
definition}, we define UEoA for bipartite
quantum states, and discuss its relation
with CoA, EoA, and Tsallis entanglement of assistance (TEoA).
In Section~\ref{Subsec: analytic}, we provide an
analytic lower bound of UEoA in two-qubit systems.
In Section~\ref{Sec: poly}, we
derive a class of polygamy inequalities of multi-qubit entanglement in terms of
UEoA, and summarize our results in Section~\ref{Conclusion}.


\section{Unified Entanglement and Unified Entanglement of Assistance}
\label{Sec: Tqentanglement}

\subsection{Definition}
\label{Subsec: definition}
Let us first recall the definition of unified entropy for quantum states~\cite{ue1, ue2}.
For $q,~s \geq 0$ such that $q \neq 1$, $s \neq 0$, unified-$(q,s)$ entropy of a
quantum state $\rho$ is
\begin{equation}
S_{q,s}(\rho):=\frac{1}{(1-q)s}\left[{\left(\T \rho^{q}\right)}^s-1\right].
\label{uqs-entropy}
\end{equation}
Unified-$(q,s)$ entropy has singularities at $q=1$ or $s=0$, however
it converges to von Neumann entropy as $q$ tends to 1;
\begin{align}
\lim_{q \rightarrow 1}S_{q,s}(\rho)=-\T \rho\log\rho 
=:S(\rho), \label{T1}
\end{align}
and R\'enyi-$q$ entropy~\cite{renyi, horo} as $s$ tends to $0$,
\begin{equation}
\lim_{s \rightarrow 0}S_{q,s}(\rho)=\frac{1}{1-q}\log \T \rho^{q}=:R_{q}(\rho).
\label{Renyi}
\end{equation}
For this reason, we can consider unified-$(q,s)$ entropy as von Neumann entropy
or R\'enyi-$q$ entropy
when $q=1$ or $s=0$ respectively; for any quantum state $\rho$ we just denote
$S_{1,s}(\rho)=S(\rho)$ and $S_{q,0}(\rho)=R_{q}(\rho)$.
We also note that unified-$(q,s)$ entropy converges to Tsallis-$q$ entropy ~\cite{tsallis}
when $s$ tends to $1$,
\begin{equation}
S_{q,1}(\rho)=\frac{1}{1-q}\left(\T \rho^{q}-1\right)=:T_{q}(\rho).
\label{Tsallis}
\end{equation}

For a bipartite pure state $\ket{\psi}_{AB}$ and each $q,~s \geq 0$,
its unified-$(q,s)$ entanglement~\cite{KSU} is defined as
\begin{equation}
E_{q,s}\left(\ket{\psi}_{AB} \right):=S_{q,s}(\rho_A),
\label{UEpure}
\end{equation}
where $\rho_A=\T _{B} \ket{\psi}_{AB}\bra{\psi}$ is the reduced
density matrix of $\ket{\psi}_{AB}$ onto subsystem $A$.
For a mixed state $\rho_{AB}$, its unified-$(q,s)$ entanglement is
\begin{equation}
E_{q,s}\left(\rho_{AB} \right):=\min \sum_i p_i
E_{q,s}(\ket{\psi_i}_{AB}), \label{UEmixed}
\end{equation}
where the minimum is taken over all possible pure state
decompositions of $\rho_{AB}=\sum_{i}p_i
\ket{\psi_i}_{AB}\bra{\psi_i}$.

Due to the continuity of unified-$(q,s)$ entropy with respect to $q$
and $s$, unified-$(q,s)$ entanglement in Eq.~(\ref{UEmixed})
converges to the entanglement of formation (EoF) as $q$ tends to 1,
\begin{align}
\lim_{q\rightarrow1}E_{q,s}\left(\rho_{AB} \right)=E_{\rm f}\left(\rho_{AB} \right),
\end{align}
where $E_{\rm f}(\rho_{AB})$ is EoF of $\rho_{AB}$ defined as
\begin{equation}
E_{\rm f}(\rho_{AB})=\min \sum_{i}p_i S(\rho^{i}_{A}) \label{EoF}
\end{equation}
with $\T_{B}|\psi^i\rangle_{AB}\langle\psi^i|=\rho^{i}_{A}$ and the
minimization being taken over all possible pure state decompositions of $\rho_{AB}=\sum_{i}p_i
\ket{\psi_i}_{AB}\bra{\psi_i}$.
When $s$
tends to $0$, unified-$(q,s)$ entanglement reduces to a
one-parameter class of entanglement measures namely R\'enyi-$q$
entanglement~\cite{KSRenyi}
\begin{align}
\lim_{s\rightarrow 0}E_{q,s}\left(\rho_{AB} \right)={\mathcal
R}_{q}\left(\rho_{AB} \right).
\label{unirenyi}
\end{align}
Unified-$(q,s)$ entanglement also reduces to another one-parameter
class called Tsallis-$q$ entanglement~\cite{KT} as $s$ tends to $1$,
\begin{align}
\lim_{s\rightarrow 1}E_{q,s}\left(\rho_{AB} \right)={\mathcal
T}_{q}\left(\rho_{AB} \right).
\label{uniT}
\end{align}
In other words, unified-$(q,s)$ entanglement is a two-parameter
generalization of EoF including the classes of R\'enyi and Tsallis
entanglement as special cases.

As a dual concept of EoF, EoA of a bipartite mixed state $\rho_{AB}$ is defined as~\cite{cohen}
\begin{equation}
E^a(\rho_{AB})=\max \sum_{i}p_i S(\rho^{i}_{A}), \label{EoA}
\end{equation}
where the maximum is taken over all possible pure state
decompositions of $\rho_{AB}=\sum_{i} p_i
|\psi^i\rangle_{AB}\langle\psi^i|$ with
$\T_{B}|\psi^i\rangle_{AB}\langle\psi^i|=\rho^{i}_{A}$. Here, we
note that EoA in Eq.~(\ref{EoA}) is clearly a mathematical dual to
EoF in Eq.~(\ref{EoF}) because one is the maximum average
entanglement over all possible pure state decompositions whereas the
other takes the minimum. Moreover, by introducing a third party $C$
that has the purification of $\rho_{AB}$, $E^a(\rho_{AB})$ can also
be considered as the maximum achievable entanglement between $A$ and
$B$ assisted by $C$~\cite{BGK}. (This is the reason why it is called the {\em
assistance}.) In other words, $E^a(\rho_{AB})$ is the maximal
entanglement that can be distributed between $A$ and $B$
assisted by the environment $C$; therefore, EoA is also physically dual
to the concept of {\em formation}.

Similar to the duality between EoF and EoA, we define UEoA of $\rho_{AB}$
as the maximum average entanglement
\begin{equation}
E_{q,s}\left(\rho_{AB} \right):=\max \sum_i p_i
E_{q,s}(\ket{\psi_i}_{AB}) \label{UEoA}
\end{equation}
over all possible pure state decompositions of $\rho_{AB}$.
Due to the continuity of unified entropy with respect to $q$
and $s$, we have
\begin{align}
\lim_{q\rightarrow1}{E^a_{q,s}}\left(\rho_{AB}
\right)=E^a\left(\rho_{AB} \right), \label{UE_EoA}
\end{align}
where $E^a(\rho_{AB})$ is the EoA of $\rho_{AB}$ in Eq.~(\ref{EoA}).
When $q$ tends to $1$ UEoA reduces to TEoA~\cite{KT},
\begin{align}
\lim_{s\rightarrow 1}E^a_{q,s}\left(\rho_{AB} \right)={\mathcal
T}^a_{q}\left(\rho_{AB} \right),
\nonumber\\
\label{uniToA}
\end{align}
where ${\mathcal T}^a_{q}\left(\rho_{AB} \right)$ is TEoA of
$\rho_{AB}$ defined as
\begin{equation}
{\mathcal T}^a_{q}\left(\rho_{AB} \right):=\max \sum_i p_i {\mathcal
T}_{q}(\ket{\psi_i}_{AB}). \label{TEoA}
\end{equation}

\subsection{Analytic Evaluation }
\label{Subsec: analytic}

For a bipartite pure state $\ket \psi_{AB}$, its concurrence~\cite{ww},
$\mathcal{C}(\ket \psi_{AB})$ is
\begin{equation}
\mathcal{C}(\ket \psi_{AB})=\sqrt{2(1-\T\rho^2_A)}, \label{pure
state concurrence}
\end{equation}
where $\rho_A=\T_B(\ket \psi_{AB}\bra \psi)$. For a mixed
state $\rho_{AB}$, its concurrence is
\begin{equation}
\mathcal{C}(\rho_{AB})=\min \sum_k p_k \mathcal{C}({\ket
{\psi_k}}_{AB}), \label{mixed state concurrence}
\end{equation}
where the minimum is taken over all possible pure state
decompositions, $\rho_{AB}=\sum_kp_k{\ket {\psi_k}}_{AB}\bra
{\psi_k}$.

For a two-qubit pure state $\ket{\psi}_{AB}$ with Schmidt
decomposition
\begin{equation}
\ket{\psi}_{AB}=\sqrt{\lambda_0}\ket{00}_{AB}+\sqrt{\lambda_1}\ket{11}_{AB}
\label{schm}
\end{equation}
with $\rho_A=\T_B(\ket \psi_{AB}\bra
\psi)=\lambda_0\ket{0}_{A}\bra{0}+\lambda_1\ket{1}_{A}\bra{1}$,
$\mathcal{C}(\ket \psi_{AB})$ in Eq.~(\ref{pure state concurrence})
can be rewritten as
\begin{equation}
\mathcal{C}(\ket
\psi_{AB})=\sqrt{2(1-\T\rho^2_A)}=2\sqrt{\lambda_0\lambda_1},
\label{conlam}
\end{equation}
Here we note that
\begin{equation}
2\sqrt{\lambda_0\lambda_1}=\left(\T\sqrt{\rho_A}\right)^2-1=
S_{\frac{1}{2}, 2}\left(\rho_A\right)=E_{\frac{1}{2},
2}\left(\ket{\psi}_{AB}\right), \label{lamuni}
\end{equation}
therefore unified-$(q,s)$ entanglement of a two-qubit pure state
$\ket{\psi}_{AB}$ reduces to the concurrence when $q=1/2$ and $s=2$.
Consequently, we have
\begin{equation}
\mathcal{C}(\rho_{AB})=E_{\frac{1}{2}, 2}\left(\rho_{AB}\right),
\label{conuni}
\end{equation}
for a two-qubit mixed state $\rho_{AB}$ because both concurrence and
unified-$(q,s)$ entanglement of bipartite mixed states are defined
by the minimum average
entanglement over all possible pure-state decompositions of
$\rho_{AB}$.

In two-qubit systems, concurrence has an analytic formula~\cite{ww};
for a two-qubit state $\rho_{AB}$,
\begin{equation}
\mathcal{C}(\rho_{AB})=\max\{0,
\lambda_1-\lambda_2-\lambda_3-\lambda_4\}, \label{C_formula}
\end{equation}
where $\lambda_i$'s are the eigenvalues, in decreasing order, of
$\sqrt{\sqrt{\rho_{AB}}\tilde{\rho}_{AB}\sqrt{\rho_{AB}}}$ and
$\tilde{\rho}_{AB}=\sigma_y \otimes\sigma_y
\rho^*_{AB}\sigma_y\otimes\sigma_y$ with the Pauli operator
$\sigma_y$. Moreover, concurrence in two-qubit systems is related
with EoF by a monotonically increasing, convex function,
\begin{equation}
 E_{\rm f} (\rho_{AB}) = {\mathcal E}(\mathcal{C}\left(\rho_{AB}\right)),
\label{C_EoF}
\end{equation}
where
\begin{equation}
{\mathcal E}(x) = H\left(\frac{1-\sqrt{1-x^2}}{2}\right),
\hspace{0.5cm}\mbox{for } 0 \le x \le 1, \label{eps}
\end{equation}
with the binary entropy function $H(t) = -[t\log t + (1-t)\log
(1-t)]$~\cite{ww}. This function relation between concurrence and
EoF is also true for any bipartite pure state with Schmidt-rank 2.
In other words, the analytic formula of concurrence in
Eq.~(\ref{C_formula}) together with the functional relation in
Eq.~(\ref{C_EoF}) lead to an analytic formula of EoF in two-qubit
systems.

Recently, it was shown that concurrence also has a functional relation with
unified-(q,s) entanglement in two-qubit systems~\cite{KSU}; for any
two-qubit mixed state $\rho_{AB}$ (as well as any bipartite pure
state with Schmidt-rank 2),
\begin{equation}
E_{q,s}\left(\rho_{AB}\right)=f_{q,s}\left(\mathcal{C}(\rho_{AB})
\right), \label{eucmix}
\end{equation}
for $q\geq1$, $0 \leq s \leq1$ and $qs\leq 3$ where $f_{q,s}(x)$ is a
differentiable function
\begin{align}
f_{q,s}(x)=&\frac{\left(\left(1+\sqrt{1-x^2}\right)^{q}
+\left(1-\sqrt{1-x^2}\right)^{q}\right)^s}{(1-q)s2^{qs}}\nonumber\\
&-\frac{1}{(1-q)s}\label{f}
\end{align}
on $0 \leq x \leq 1$. This functional relation in Eq.~(\ref{eucmix})
was established by showing the monotonicity and convexity of
$f_{q,s}(x)$ for $q\geq1$, $0 \leq s \leq1$ and $qs\leq 3$. $f_{q,s}(x)$ reduces to
${\mathcal E}(x)$ in Eq.~(\ref{eps}) as $q$ tends to 1.

Here, we note that $f_{q,s}(x)$ in Eq.~(\ref{f}) also relates UEoA
with CoA in two-qubit systems.
\begin{Lem}
For $q\geq1$, $0 \leq s \leq1$, $qs\leq 3$ and any two-qubit state $\rho_{AB}$,
\begin{equation}
E^a_{q,s}\left(\rho_{AB}\right)\geq
f_{q,s}\left(\mathcal{C}^a(\rho_{AB}) \right) \label{UEoA_CoA}
\end{equation}
where $E^a_{q,s}\left(\rho_{AB}\right)$ and
$\mathcal{C}^a(\rho_{AB})$ are UEoA and CoA of $\rho_{AB}$
respectively. \label{Lem: UEoA_CoA}
\end{Lem}

\begin{proof}
Let $\rho_{AB}=\sum_kp_k{\ket {\psi_k}}_{AB}\bra {\psi_k}$ be the
optimal decomposition realizing CoA,
\begin{align}
\mathcal{C}^a(\rho_{AB})=\sum_kp_k \mathcal{C}\left({\ket
{\psi_k}}_{AB}\right),
\end{align}
then we have
\begin{align}
f_{q,s}\left(\mathcal{C}^a(\rho_{AB})
\right)=&f_{q,s}\left(\sum_kp_k \mathcal{C}\left({\ket
{\psi_k}}_{AB}\right)\right)\nonumber\\
\leq&\sum_kp_k f_{q,s}\left(\mathcal{C}\left({\ket
{\psi_k}}_{AB}\right)\right)\nonumber\\
=&\sum_kp_k E_{q,s}\left({\ket
{\psi_k}}_{AB}\right)\nonumber\\
\leq&E^a_{q,s}\left(\rho_{AB}\right),
\end{align}
where the first inequality is due to the convexity of $f_{q,s}$ for the range of
$q\geq1$, $0 \leq s \leq1$ and $qs\leq 3$, the
second equality is the functional relation of UEoA and concurrence
for two-qubit pure states, and the last inequality is by the
definition of UEoA.
\end{proof}

Thus, together with the analytic formula of two-qubit concurrence in
Eq.~(\ref{C_formula}), Lemma~\ref{Lem: UEoA_CoA} provides an
analytic lowerbound of UEoA for two-qubit systems.

\section{Multi-qubit Polygamy Inequality of Entanglement}
\label{Sec: poly}

Using the square of concurrence (sometimes, referred to as tangle)
to quantify bipartite entanglement, monogamy of multi-qubit
entanglement was mathematically characterized as an
inequality~\cite{ckw,ov}; for an $n$-qubit pure state
$\ket{\psi}_{A_1A_2\cdots A_n}$,
\begin{equation}
\mathcal{C}_{A_1 (A_2 \cdots A_n)}^2  \geq  \mathcal{C}_{A_1 A_2}^2
+\cdots+\mathcal{C}_{A_1 A_n}^2, \label{nCmono}
\end{equation}
where $\mathcal{C}_{A_1 (A_2 \cdots
A_n)}=\mathcal{C}(\ket{\psi}_{A_1(A_2\cdots A_n)})$ is the
concurrence of $\ket{\psi}_{A_1A_2\cdots A_n}$ with respect to the
bipartite cut between $A_1$ and the others, and
$\mathcal{C}_{A_1A_i}=\mathcal{C}(\rho_{A_1A_i})$ is the concurrence
of the reduced density matrix $\rho_{A_1A_i}$ for $i=2,\ldots, n$.
This monogamous property of multi-qubit entanglement was also
established in terms of various entanglement measures using R\'enyi
and Tsallis entropies~\cite{KSRenyi, KT}, and these classes of
monogamy inequalities were recently generalized as a generic
two-parameter class in terms of unified-$(q,s)$
entanglement~\cite{KSU}.

Whereas monogamy of multipartite entanglement reveals the restricted
shareability of multi-party entanglement in terms of entanglement measures,
entanglement of assistance, the dual concept of entanglement measures,
was also shown to have a dually monogamous
(that is, polygamous) relation in multi-party quantum systems; for a
multi-qubit pure state $\ket{\psi}_{A_1 \cdots A_n}$, we have the
following polygamy inequality,
\begin{equation}
\mathcal{C}_{A_1 (A_2 \cdots A_n)}^2  \leq  (\mathcal{C}^a_{A_1 A_2})^2
+\cdots+(\mathcal{C}^a_{A_1 A_n})^2,
\label{nCdual}
\end{equation}
where $\mathcal{C}^a_{A_1 A_i}$ is the CoA of the reduced density
matrix $\rho_{A_1A_i}$ for $i=2,\ldots, n$.

In other words, the bipartite entanglement between $A_1$ and $A_2\cdots A_n$
is an upper bound for the sum of two-qubit entanglement between $A_1$ and each of $A_i's$
in monogamy inequalities. Moreover, the same quantity also plays as a lowerbound for the sum of
two-qubit distributed entanglement in the polygamy inequality.
For three-party pure states, a polygamy inequality of entanglement
was also introduced by using EoA~\cite{BGK}, and
a class of polygamy inequalities for multi-qubit mixed states
was also introduced using TEoA~\cite{KT}.

Here we establish a unified view of this
polygamous property of multi-qubit entanglement by introducing a
two-parameter class of polygamy inequalities in terms of UEoA.
Before we provide the class of polygamy inequalities,
we first prove an important property of the function $f_{q,s}(x)$
in Eq.~(\ref{f}).

\begin{Lem}
For $1 \leq q\leq 2$ and $-q^2+4q-3 \leq s \leq1$,
\begin{align}
f_{q,s}\left(\sqrt{x^2+y^2}\right)\leq f_{q,s}(x)+f_{q,s}(y)
\label{fnega}
\end{align}
for $0\leq x, y, x^2+y^2 \leq1$.
 \label{Lem: fnega}
\end{Lem}

\begin{proof}
In fact, Inequality~(\ref{fnega}) was already shown when $q=1$ or
$q=2$ (consequently $s=1$)~\cite{BGK, KT} so we prove the lemma for
the case of $1<q<2$. The proof method follows the construction used
in~\cite{KSU}.

For $1 < q < 2$ and $-q^2+4q-3 \leq s \leq1$, let us define a
two-variable function $h_q(x,y)$,
\begin{equation}
h_{q,s}(x,y):=f_{q,s}\left(\sqrt{x^2+y^2}\right)-f_{q,s}(x)-f_{q,s}(y),
\label{m_q}
\end{equation}
on the domain ${\mathcal D}=\{ (x,y)| 0\leq x, y, x^2+y^2 \leq1\}$,
then Inequality~(\ref{fnega}) is equivalent to show that $h_{q,s}(x,y)\leq 0$ for the range of
$q$ and $s$.

Because $h_{q,s}(x,y)$ is continuous on the domain $\mathcal D$ and
differentiable in the interior $\mathring {\mathcal D}$, its maximum
or minimum values can arise only at the critical points or on the
boundary of $\mathcal D$. The gradient of $h_{q,s}(x,y)$ is
\begin{equation}
\nabla h_{q,s}(x, y)=\left(\frac{\partial h_{q,s}(x,y)}{\partial x},
\frac{\partial h_{q,s}(x,y)}{\partial y}\right) \label{grad}
\end{equation}
where the first-order partial derivatives are

\begin{widetext}
\begin{align}
\frac{\partial h_{q,s}(x,y)}{\partial x}= &\Gamma
\frac{qsx}{\sqrt{1-x^2}}\left(\Theta(x)^{q}+{\Xi(x)}^{q}\right)^{s-1}
\left({\Theta(x)}^{q-1}-{\Xi(x)}^{q-1}\right)\nonumber\\
&-\Gamma
\frac{qsx}{\sqrt{1-x^2-y^2}}\left({\Theta\left(\sqrt{x^2+y^2}\right)}^{q}
+{\Xi\left(\sqrt{x^2+y^2}\right)}^{q}\right)^{s-1}
\left({\Theta\left(\sqrt{x^2+y^2}\right)}^{q-1}-{\Xi\left(\sqrt{x^2+y^2}\right)}^{q-1}\right),\nonumber\\
\frac{\partial h_{q,s}(x,y)}{\partial y}= &\Gamma
\frac{qsy}{\sqrt{1-y^2}}\left(\Theta(y)^{q}+{\Xi(y)}^{q}\right)^{s-1}
\left({\Theta(y)}^{q-1}-{\Xi(y)}^{q-1}\right)\nonumber\\
&-\Gamma
\frac{qsy}{\sqrt{1-x^2-y^2}}\left({\Theta\left(\sqrt{x^2+y^2}\right)}^{q}
+{\Xi\left(\sqrt{x^2+y^2}\right)}^{q}\right)^{s-1}
\left({\Theta\left(\sqrt{x^2+y^2}\right)}^{q-1}-{\Xi\left(\sqrt{x^2+y^2}\right)}^{q-1}\right)
\label{2pderi}
\end{align}
\end{widetext}
with $\Theta(t)=1+\sqrt{1-t^2}$, $\Xi(t) =1-\sqrt{1-t^2}$ and
$\Gamma=1/\left[(1-q)s2^{sq}\right]$.

Suppose that there exists $(x_0, y_0) \in \mathring {\mathcal D}$
such that $\nabla h_{q,s}(x_0, y_0)=(0,0)$, then
Eq.~(\ref{2pderi}) implies
\begin{equation}
n_{q,s}(x_0)=n_{q,s}(y_0), \label{x0y0}
\end{equation}
where $n_{q,s}(t)$ is a differentiable function
\begin{align}
n_{q,s}(t)=\frac{qs}{\sqrt{1-t^2}}&\left({\Theta(t)}^{q}+{\Xi(t)}^{q}\right)^{s-1}\nonumber\\
&\cdot\left({\Theta(t)}^{q-1}-{\Xi(t)}^{q-1}\right), \label{n_q}
\end{align}
on $0<t<1$.

We first show that $n_{q,s}(t)$ is a strictly increasing function
and thus Eq.~(\ref{x0y0}) implies $x_0=y_0$. This is also enough to
show that ${{\rm d}n_{q,s}(t)}/{{\rm d}t}>0$ for $0<t<1$ because
$n_{q,s}(t)$ is differentiable with respect to $t$. The first-order
derivative of $n_{q,s}(t)$ is
\begin{align}
\frac{{\rm d}n_{q,s}(t)}{{\rm d}t}=&
\Omega\frac{t^2\left(\Theta(t)^q+\Xi(t)^q\right)\left(\Theta(t)^{q-1}-\Xi(t)^{q-1}\right)}
{\sqrt{1-t^2}}\nonumber\\
&+\Omega q(1-s)t^2\left(\Theta(t)^{q-1}-\Xi(t)^{q-1}\right)^2\nonumber\\
&-\Omega(q-1)t^2
\left(\Theta(t)^q+\Xi(t)^q\right)\nonumber\\
&~~~~~~~~~~~\cdot\left(\Theta(t)^{q-2}+\Xi(t)^{q-2}\right)
\label{g2dn}
\end{align}
with $\Omega=qs\left(\Theta(t)^q+\Xi(t)^q \right)^{s-2}/t\left(
1-t^2\right)$.

For $1 < q < 2$ and $-q^2+4q-3 \leq s \leq1$, we have $q(1-s)\geq
q-2 > q-3$, thus
\begin{widetext}
\begin{align}
\frac{{\rm d}n_{q,s}(t)}{{\rm d}t}>& \Omega
t^2\left(\Theta(t)^q+\Xi(t)^q\right)\left[\frac{\left(\Theta(t)^{q-1}-\Xi(t)^{q-1}\right)}{\sqrt{1-t^2}}
-2\left(\Theta(t)^{q-2}+\Xi(t)^{q-2}\right)\right]\nonumber\\
&+\Omega(q-3)t^2\left[\left(\Theta(t)^{q-1}-\Xi(t)^{q-1}\right)^2
-\left(\Theta(t)^q+\Xi(t)^q\right)\left(\Theta(t)^{q-2}+\Xi(t)^{q-2}\right)\right].
\label{g2dn2}
\end{align}
\end{widetext}
Due to the relation $\Theta(t)-\Xi(t)=2\sqrt{1-t^2}$, we have
\begin{align}
\frac{\left(\Theta(t)^{q-1}-\Xi(t)^{q-1}\right)}{\sqrt{1-t^2}}
-&2\left(\Theta(t)^{q-2}+\Xi(t)^{q-2}\right)\nonumber\\
=&\frac{t^2\left(\Theta(t)^{q-3}-\Xi(t)^{q-3}\right)}{\sqrt{1-t^2}},
\label{g2dn2rel}
\end{align}
and the binomial series of $\Theta(t)^{\alpha}=\left(1+\sqrt{1-t^2}
\right)^{\alpha}$ and $\Xi(t)^{\alpha}=\left(1-\sqrt{1-t^2}
\right)^{\alpha}$ lead us to
\begin{align}
\Theta(t)^{\alpha}-\Xi(t)^{\alpha}
 \geq & 2\alpha\sqrt{1-t^2},~
 \Theta(t)^{\alpha}+\Xi(t)^{\alpha}
\geq 2 \label{lower}
\end{align}
for real $\alpha$. Furthermore, using the relations
$\Theta(t)+\Xi(t)=2$ and $\Theta(t)\Xi(t)=t^2$, it is also
straightforward to verify that
\begin{align}
&\left(\Theta(t)^{q-1}-\Xi(t)^{q-1}\right)^2\nonumber\\
&~~~~~~~-\left(\Theta(t)^q+\Xi(t)^q\right)
\left(\Theta(t)^{q-2}+\Xi(t)^{q-2}\right)=t^{2q-4}.
\label{squarsimple}
\end{align}
Thus, together with Eqs.~(\ref{g2dn2rel}), (\ref{lower}) and
(\ref{squarsimple}), Inequality~(\ref{g2dn2}) yields
\begin{align}
\frac{{\rm d}n_{q,s}(t)}{{\rm d}t} >& \Omega
\frac{t^4\left(\Theta(t)^q+\Xi(t)^q\right)\left(\Theta(t)^{q-3}-\Xi(t)^{q-3}\right)}
{\sqrt{1-t^2}}\nonumber\\
&~~~~-\Omega4(q-3)t^{2q-2}\nonumber\\
\geq& 4\Omega(q-3)(t^4-t^{2q-2}). \label{g2dn3}
\end{align}

The last term of the inequality is strictly positive for $1< q < 2$
and $0<t<1$, therefore $n_{q,s}(t)$ is a strictly increasing
function for $1< q< 2$ and $-q^2+4q-3 \leq s \leq1$. In other words,
Eq.~(\ref{x0y0}) implies $x_0=y_0$. However, from
Eq.~(\ref{2pderi}), $\nabla h_{q,s}(x_0,
x_0)=(0,0)$ also implies that $n_{q,s}(x_0)=n_{q,s}(\sqrt{2}x_0)$
for some $x_0 \in (0,1)$, which contradicts the strict monotonicity
of $n_{q,s}(t)$. Thus $h_{q,s}(x, y)$ does not have any vanishing
gradient in $\mathring {\mathcal D}$ for $1< q< 2$ and $-q^2+4q-3
\leq s \leq1$.

Now let us consider the function value of $h_{q,s}(x, y)$ on the
boundary of $\mathcal D$, that is, either $x=0$ or $y=0$ or
$x^2+y^2=1$. If $x=0$ or $y=0$ , then clearly $h_{q,s}(x, y)=0$.
Suppose $x^2+y^2=1$ with $x\neq 0$ and $y\neq0$. Then $h_{q,s}(x,
y)$ becomes a single-variable function,
\begin{align}
l_{q,s}(x)
:=&\frac{\left(\left(1+\sqrt{1-x^2}\right)^{q}+\left(1-\sqrt{1-x^2}\right)^{q}\right)^{s}
}{(q-1)s2^{qs}}\nonumber\\
&+\frac{\left(\left(1+x\right)^{q}+\left(1-x\right)^{q}\right)^{s}-2^s-2^{qs}}{(q-1)s2^{qs}}
\label{lqs}
\end{align}
for $0< x< 1$.

Because $(q-1)s2^{qs}>0$ for $1 < q < 2$ and $-q^2+4q-3 \leq s
\leq1$, the sign of the function $l_{q,s}(x)$ is same with that of
the following differentiable function
\begin{align}
m_{q,s}(x):=&\left(\left(1+\sqrt{1-x^2}\right)^{q}+\left(1-\sqrt{1-x^2}\right)^{q}\right)^{s}\nonumber\\
&+\left(\left(1+x\right)^{q}+\left(1-x\right)^{q}\right)^{s}-2^s-2^{qs}.
\label{mqs}
\end{align}
If we consider the derivative of $m_{q,s}(x)$,
\begin{align}
\frac{{\rm d}m_{q,s}(x)}{{\rm
d}x}=&sq\left[\left(1+x\right)^{q}+\left(1-x\right)^{q}\right]^{s-1}\nonumber\\
&\cdot\left[\left(1+x\right)^{q-1}-\left(1-x\right)^{q-1}\right]\nonumber\\
&-\frac{sqx \left[\left(1+\sqrt{1-x^2}\right)^{q}+\left(1-\sqrt{1-x^2}\right)^{q}\right]^{s-1}}{\sqrt{1-x^2}}\nonumber\\
&\cdot\left[\left(1+\sqrt{1-x^2}\right)^{q-1}-\left(1-\sqrt{1-x^2}\right)^{q-1}\right],
\label{mqsderi}
\end{align}
we note that $x=1/\sqrt{2}$ is the only critical point of
$m_{q,s}(x)$ on $0 < x <1$. Furthermore, it is also straightforward to
verify that $m_{q,s}\left(1/\sqrt{2}\right)\leq 0$ for $1 < q < 2$
and $-q^2+4q-3 \leq s \leq1$, which is illustrated in
Figure~\ref{fig1}.
Because $m_{q,s}(0)=m_{q,s}(1)=0$ and
$m_{q,s}\left(1/\sqrt{2}\right)\leq 0$ where $x=1/\sqrt{2}$ is the
only critical point of $m_{q,s}(x)$, $m_{q,s}(x)\leq 0$ through out
the whole range of $0\leq x \leq1$. In other words,
$h_{q,s}(x,y)\leq 0$ for $1 < q < 2$ and $-q^2+4q-3 \leq s \leq1$,
which complete the proof.
\end{proof}
\begin{figure}
\includegraphics[width=.9\linewidth]{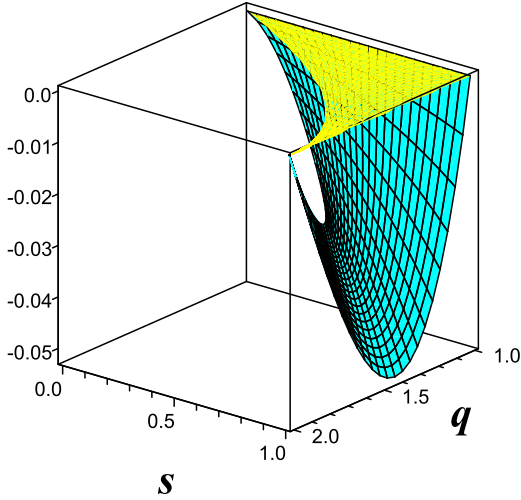}\\
\caption{(Color online) The function values of
$m_{q,s}\left(1/\sqrt{2}\right)$ (patched blue curved surface) are
indicated for real parameters $s$ and $q$. Yellow area on the top of the box
indicates the range of $q$ and $s$ such that $1 \leq q \leq 2$ and $-q^2+4q-3 \leq
s \leq1$.} \label{fig1}
\end{figure}
Now, we are ready to have the following theorem about the polygamy
of multi-qubit entanglement using unified-$(q,s)$ entropy.
\begin{Thm}
For $1 \leq q \leq 2$, $-q^2+4q-3 \leq s \leq1$ and any multi-qubit
state $\rho_{A_1 \cdots A_n}$ , we have
\begin{equation}
E^a_{q,s}\left( \rho_{A_1(A_2 \cdots A_n)}\right)\leq
E^a_{q,s}(\rho_{A_1 A_2}) +\cdots+E^a_{q,s}(\rho_{A_1 A_n})
\label{UEpoly}
\end{equation}
where $E_{q,s}\left( \rho_{A_1(A_2 \cdots A_n)}\right)$ is the
unified-$(q,s)$ entanglement of $\rho_{A_1A_2 \cdots A_n}$ with
respect to the bipartition between $A_1$ and $A_{2}\cdots A_{n}$,
and $E^a_{q,s}(\rho_{A_1 A_i})$ is the UEoA of the reduced density
matrix $\rho_{A_1 A_i}$ for $i=2,\cdots,n$. \label{Thm: poly}
\end{Thm}
\begin{proof}
We first prove the theorem for a $n$-qubit pure states, and
generalize the proof into mixed states. For a $n$-qubit pure state
$\ket{\psi}_{A_1(A_2\cdots A_n)}$, let us first assume that $
\mathcal{C}_{A_1 (A_2 \cdots A_n)}^2 \leq(\mathcal{C}^a_{A_1 A_2})^2
+(\mathcal{C}^a_{A_1 A_3})^2+\cdots+(\mathcal{C}^a_{A_1 A_n})^2 \leq
1$ in Eq.~(\ref{nCdual}). Then we have
\begin{align}
E_{q,s}\left(\ket{\psi}_{A_1(A_2 \cdots A_n)}\right)
&=f_{q,s}({\mathcal C}_{A_1(A_2 \cdots A_n)})\nonumber\\
&\leq f_{q,s}\left(\sqrt{(\mathcal{C}^a_{A_1 A_2})^2
+\cdots+(\mathcal{C}^a_{A_1 A_n})^2 }\right)\nonumber\\
&\leq f_{q,s}\left( \mathcal{C}^a_{A_1 A_2}\right)\nonumber\\
&~+f_{q,s}\left(\sqrt{(\mathcal{C}^a_{A_1 A_3})^2
+\cdots+(\mathcal{C}^a_{A_1 A_n})^2}\right)\nonumber\\
&~~~~~~~\vdots\nonumber\\
&\leq f_{q,s}\left( \mathcal{C}^a_{A_1 A_2}\right)+
+\cdots+f_{q,s}\left(
\mathcal{C}^a_{A_1 A_n}\right)\nonumber\\
&\leq E_{q,s}^a\left(\rho_{A_1 A_2}\right)
+\cdots+E_{q,s}^a\left(\rho_{A_1 A_n}\right),
\end{align}
where the first inequality is due to the monotonicity of the
function $f_{q,s}(x)$, the second and third inequalities are
obtained by iterative use of Lemma~\ref{Lem: fnega}, and the last
inequality is by Lemma~\ref{Lem: UEoA_CoA}.

Now, let us assume that $ \mathcal{C}_{A_1 (A_2 \cdots
A_n)}^2\leq1<(\mathcal{C}^a_{A_1 A_2})^2 +\cdots+(\mathcal{C}^a_{A_1
A_n})^2 $. Because $f_{q,s}(x)$ is an increasing
function~\cite{KSU}, we have
\begin{align}
E_{q,s}\left(\ket{\psi}_{A_1(A_2 \cdots A_n)}\right)=&
f_{q,s}\left(\mathcal{C}_{A_1 (A_2 \cdots A_n)}^2\right)\nonumber\\
\leq& f_{q,s}\left(1\right)
\end{align}
for any multi-qubit pure state $\ket{\psi}_{A_1(A_2 \cdots A_n)}$.
Thus it is enough to show that $E_{q,s}^a(\rho_{A_1 A_2})
+\cdots+E_{q,s}^a(\rho_{A_1 A_n}) \geq f_{q,s}\left(1\right)$.

Our assumption $1<(\mathcal{C}^a_{A_1 A_2})^2
+\cdots+(\mathcal{C}^a_{A_1 A_n})^2 $ implies that there exists $k
\in \{2,\ldots ,n-1 \}$ such that
\begin{align}
&(\mathcal{C}^a_{A_1 A_2})^2 +\cdots+(\mathcal{C}^a_{A_1 A_k})^2
\leq 1,\nonumber\\
&(\mathcal{C}^a_{A_1 A_2})^2 +\cdots+(\mathcal{C}^a_{A_1 A_{k+1}})^2
>1.
\end{align}
By letting
\begin{equation} T:=(\mathcal{C}^a_{A_1 A_2})^2
+\cdots+(\mathcal{C}^a_{A_1 A_{k+1}})^2-1>0,
\end{equation}
we have
\begin{align}
f_{q,s}\left(1\right) =&f_{q,s}\left(1\right)\nonumber\\
=& f_{q,s} \left( \sqrt{(\mathcal{C}^a_{A_1 A_2})^2
+\cdots+(\mathcal{C}^a_{A_1
A_{k+1}})^2-T} \right)\nonumber\\
\leq& f_{q,s} \left( \sqrt{(\mathcal{C}^a_{A_1 A_2})^2
+\cdots+(\mathcal{C}^a_{A_1 A_k})^2} \right)\nonumber\\
&~~~~~~+f_{q,s} \left( \sqrt{(\mathcal{C}^a_{A_1 A_{k+1}})^2-T} \right)\nonumber\\
\leq& f_{q,s} \left( \mathcal{C}^a_{A_1 A_2}\right)+\cdots +f_{q,s}
\left( \mathcal{C}^a_{A_1 A_k}\right)+
f_{q,s} ( \mathcal{C}^a_{A_1 A_{k+1}})\nonumber\\
\leq& E_{q,s}^a(\rho_{A_1 A_2})+\cdots + E_{q,s}^a(\rho_{A_1 A_n}),
\label{nTmonopure2}
\end{align}
where the first inequality is by using Lemma~\ref{Lem: fnega} with
respect to $(\mathcal{C}^a_{A_1 A_2})^2+\cdots+(\mathcal{C}^a_{A_1
A_k})^2$ and $(\mathcal{C}^a_{A_1 A_{k+1}})^2-T$, the second
inequality is by iterative use of Lemma~\ref{Lem: fnega} on
$(\mathcal{C}^a_{A_1 A_2})^2+\cdots+(\mathcal{C}^a_{A_1 A_k})^2$,
and the last inequality is by Lemma~\ref{Lem: UEoA_CoA}.

Now let us consider multi-qubit mixed states.
For a $n$-qubit mixed state  $\rho_{A_1A_2\cdots A_n}$, let
$\rho_{A_1(A_2\cdots A_n)}=\sum_j p_j \ket{\psi_j}_{A_1(A_2\cdots
A_n)}\bra{\psi_j}$ be an optimal decomposition for UEoA such that
\begin{align}
E_{q,s}^a\left(\rho_{A_1(A_2\cdots A_n)}\right)=\sum_j p_j
E_{q,s}\left(\ket{\psi_j}_{A_1(A_2\cdots A_n)}\right).
\label{psiopt}
\end{align}
Because each
$\ket{\psi_j}_{A_1(A_2\cdots A_n)}$ in the decomposition is an
$n$-qubit pure state, we have
\begin{align}
E_{q,s}\left(\ket{\psi_j}_{A_1(A_2\cdots
A_n)}\right) \leq E_{q,s}^a\left(\rho^j_{A_1A_2}\right)+&\nonumber\\
\cdots+&E_{q,s}^a\left(\rho^j_{A_1A_n}\right)
\label{unipolypsii}
\end{align}
where $\rho^j_{A_1A_i}$ is the reduced density matrix of
$\ket{\psi_j}_{A_1\cdots A_n}$ onto two-qubit subsystem $A_1A_i$ for
each $i=2,\cdots,n$.
From Eq.~(\ref{psiopt}) together with Inequality~(\ref{unipolypsii}), we have
\begin{align}
E_{q,s}^a\left(\rho_{A_1(A_2\cdots A_n)}\right)=&\sum_j p_j
E_{q,s}\left(\ket{\psi_j}_{A_1(A_2\cdots
A_n)}\right)\nonumber\\
\leq &\sum_j p_jE_{q,s}^a\left(\rho^j_{A_1A_2}\right)+\cdots\nonumber\\
&~~~~~~~~~~ +\sum_j
p_jE_{q,s}^a\left(\rho^j_{A_1A_n}\right) \nonumber\\
\leq&E_{q,s}^a\left(\rho_{A_1A_2}\right)+\cdots
+E_{q,s}^a\left(\rho_{A_1A_n}\right),
\label{unipolymixed}
\end{align}
where the last inequality is by definition of UEoA
for each $\rho_{A_1A_i}$.
\end{proof}

We note that Inequality (\ref{UEpoly}) is reduced to Tsallis-$q$
monogamy inequality~\cite{KT}
\begin{equation}
{\mathcal T}^a_{q}\left( \rho_{A_1(A_2 \cdots A_n)}\right)\leq
{\mathcal T}^a_{q}(\rho_{A_1 A_2}) +\cdots+{\mathcal
T}^a_{q}(\rho_{A_1 A_n}) \label{Tpoly}
\end{equation}
as $s$ tends to 1, and it also reduces to the multi-qubit polygamy
inequality in terms of EoA~\cite{BGK} as $q$ tends to 1. For $q=2$
and $s=1$, unified-$(q,s)$ entanglement coincides with the squared
concurrence for two-qubit pure states; for a bipartite pure state
$\ket{\psi}_{AB}$ with Schmidt-rank 2,
\begin{equation}
E_{2,1}\left(\ket{\psi}_{AB}\right)={\mathcal
C}^2\left(\ket{\psi}_{AB}\right). \label{unitangle}
\end{equation}
For this relation, it is also straightforward to verify that
Inequality (\ref{UEpoly}) reduces to Inequality~(\ref{nCdual}) as
$q\rightarrow 2$ and $s\rightarrow1$.
Thus, Theorem~\ref{Thm: poly} provides an interpolation among EoA,
TEoA and CoA polygamy inequalities of multi-qubit entanglement,
which is illustrated in Figure~\ref{fig2}.
\begin{figure}
\includegraphics[width=.9\linewidth]{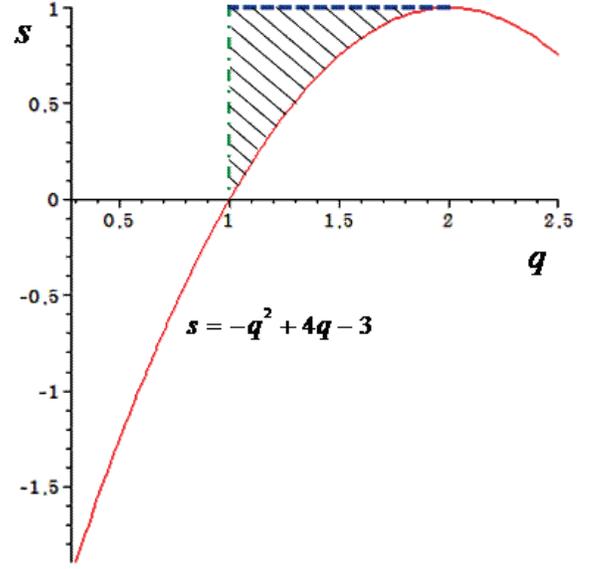}\\
\caption{(Color online) The domain of real parameters $q$ and $s$
where multi-qubit polygamy inequality holds in terms of UEoA. The
dashed line indicates the domain for which the multi-qubit polygamy
inequality holds for Tsallis-$q$ entropy (TEoA), and the dashed-dot
line is the domain for von Neumann entropy (EoA). The shaded range
is for unified-$(q,s)$ entropy (UEoA).} \label{fig2}
\end{figure}
We further note that the continuity of unified-$(q,s)$ entropy also
guarantees multi-qubit polygamy inequality in terms of UEoA when $q$
and $s$ are slightly outside of the proposed domain in
Figure~\ref{fig2}.

In three-qubit systems, Inequality~(\ref{UEpoly}) in Theorem~\ref{Thm: poly} can be improved into a
tighter form. A direct observation
from~\cite{ckw} shows
\begin{equation}
\mathcal{C}_{A(BC)}^2 = \mathcal{C}_{AB}^2 + ({\mathcal{C}^{a}_{AC}})^2,
\label{3tangle}
\end{equation}
for a 3-qubit pure state $\ket{\psi}_{ABC}$
where $\mathcal{C}_{AB}$ and $\mathcal{C}^{a}_{AC}$ are the
concurrence and CoA of $\rho_{AB}$ and $\rho_{AC}$ respectively.

From Eq.~(\ref{3tangle}) together with Lemma~\ref{Lem: fnega}, we
have the following tighter polygamy inequality of three-qubit
entanglement.
\begin{Thm}
For $1 \leq q \leq 2$, $-q^2+4q-3 \leq s \leq1$ and any three-qubit
pure state $\ket{\psi}_{ABC}$, we have
\begin{equation}
E_{q,s}\left( \ket{\psi}_{A(BC)}\right)\leq E_{q,s}(\rho_{AB})
+E^a_{q,s}(\rho_{AC}) \label{3UEpoly}
\end{equation}
where $E_{q,s}\left( \ket{\psi}_{A(BC)}\right)$ is the
unified-$(q,s)$ entanglement of $\ket{\psi}_{ABC}$ with respect to
the bipartition between $A$ and $BC$, $E_{q,s}(\rho_{AB})$ is the
unified-$(q,s)$ entanglement of $\rho_{AB}$ and
$E^a_{q,s}(\rho_{AC})$ is the UEoA of $\rho_{AC}$. \label{3unipoly}
\end{Thm}

\begin{proof}
Because $\ket{\psi}_{ABC}$ is a bipartite pure state between $A$ and
$BC$ with Schmidt-rank less than or equal to two, we have
\begin{equation}
E_{q,s}\left(\psi_{A(BC)}\right)=f_{q,s}\left(\mathcal{C}_{A(BC)}\right).
\end{equation}
Thus,
\begin{align}
f_{q,s}\left(\mathcal{C}_{A(BC)}\right)
&= f_{q,s}\left(\sqrt{\mathcal{C}^{2}_{AB}+{\mathcal{C}^{a}_{AC}}^2}\right)\nonumber\\
&\leq  f_{q,s}(\mathcal{C}_{AB})+ f_{q,s}(\mathcal{C}^{a}_{AC})\nonumber\\
&\leq  E_{q,s}(\rho_{AB})+E^a_{q,s}(\rho_{AC}),
\end{align}
where the first inequality is by Lemma~\ref{Lem: fnega}, and the
second inequality is by Lemma~\ref{Lem: UEoA_CoA}.
\end{proof}


\section{Conclusion}
\label{Conclusion}

Using unified-$(q,s)$ entropy, we have provided a two-parameter
generalization of EoA, namely UEoA with an analytical lowerbound
in two-qubit systems for $q\geq 1$, $0\leq s \leq1$ and $qs\leq3$.
Based on this unified formalism of EoA, we have established a broad class of
multi-qubit polygamy inequalities in terms of unified-$(q,s)$
entanglement for $1 \leq q \leq 2$, $-q^2+4q-3 \leq s \leq1$.
We have also shown a tighter polygamy inequality for the case of
three-qubit pure states.

The class of polygamy inequalities we provided here
encapsulates every known case of multi-qubit polygamy inequality
in terms of EoA, CoA or TEoA as special cases, as well as their explicit
relation with respect to a differential function $f_{q,s}(x)$.
Thus our result provides a useful
methodology to understand the restricted distribution of entanglement in multi-party
quantum systems.

\section*{Acknowledgments}
This work was supported by Emerging Technology R\&D Center of SK Telecom.


\end{document}